\documentclass{llncs}

\pdfoutput=1

\usepackage{times}
\usepackage{graphicx}
\usepackage{color}
\usepackage{units}
\usepackage{enumerate}
\usepackage{amsmath}
\usepackage{amssymb}
\usepackage{hhline}
\usepackage{listings}          \usepackage{paralist}
\usepackage{xcolor}
\usepackage{enumitem}
\usepackage{lscape}

\usepackage[noline,noend,ruled,linesnumbered,noresetcount]{algorithm2e}

\usepackage{hyperref}
\usepackage{cleveref}
\usepackage{fixltx2e}
\usepackage{multirow}

\spnewtheorem{observation}{Observation}{\bfseries}{\itshape}

\Crefname{observation}{Observation}{Observations}
\Crefname{section}{Sec.}{Sec.}
\Crefname{definition}{Def.}{Def.}
\Crefname{algorithm}{Alg.}{Alg.}
\Crefname{lemma}{Lem.}{Lem.}
\Crefname{theorem}{Thm.}{Thm.}
\Crefname{corollary}{Cor.}{Cor.}
\Crefname{figure}{Fig.}{Fig.}

\newif\ifproofs

\newif\iflong
\longtrue

\iflong
\proofstrue
\else
\proofsfalse
\fi

\newcommand{\remspace}{\vspace{-0.2cm}}

\newcommand{\card}[1]{{\left\vert{#1}\right\vert}} 

\newcommand{\powerset}[1]{\mathcal{P}(#1)}

\newcommand{\eqdef}{\buildrel \mbox{\tiny\rm def} \over =}

\newcommand{\true}{{\textit{true}}}
\newcommand{\false}{{\textit{false}}}

\newcommand{\letter}{{a}}

\newcommand{\struct}{s}
\newcommand{\Dom}{{\mathcal{D}}}
\newcommand{\Int}{{\mathcal{I}}}

\newcommand{\func}{f}
\newcommand{\srt}{s}

\newcommand{\stgraph}[1]{\textit{QA}(#1)}

\newcommand{\init}{I}

\newcommand{\Inv}{\textit{Inv}}

\newcommand{\TR}{\textit{TR}}

\newcommand{\safe}{P}

\newcommand{\bsort}{\sort{s}}

\newcommand{\sort}[1]{\mathsf{#1}}

\newcommand{\sbapaint}{\sort{int}}
\newcommand{\sbapaset}{\sort{set}}

\newcommand{\relation}[1]{{\textit{#1}}}

\newcommand{\snode}{\sort{node}}

\newcommand{\squorumof}[1]{\sort{set}_{\mathsf{#1}}}

\newcommand{\svalue}{\sort{value}}

\newcommand{\rmember}{\relation{member}}

\newcommand{\rdecision}{\relation{decision}}
\newcommand{\rundcons}{\relation{und\_cons}}

\newcommand{\rrcvmsg}{\relation{rcv\_msg}}

\newcommand{\remove}[1]{\xspace}

\newcommand{\allnotes}[1]{}

\newcommand{\para}[1]{\vspace{3pt}{\em #1.}}

\newcommand{\Param}{\textit{Prm}}
\newcommand{\ParamInt}{\Param_{I}}
\newcommand{\ParamSet}{\Param_{S}}

\newcommand{\paramof}[1]{{\bf #1}}
\newcommand{\nP}{\paramof{n}}
\newcommand{\fP}{\paramof{f}}
\newcommand{\tP}{\paramof{t}}

\newcommand{\ToBAPA}{\mathcal{B}}

\newcommand{\structbapa}{\struct_B}
\newcommand{\Bterm}{b}
\newcommand{\assmp}{\Gamma}
\newcommand{\cSet}{{\bf a}}
\newcommand{\cInt}{i}
\newcommand{\Bgrammar}{B}
\newcommand{\Kgrammar}{K}

\newcommand{\Lgrammar}{L}
\newcommand{\Fgrammar}{F}

\newcommand{\Val}{\iota}

\newcommand{\ToFOL}{\mathcal{FO}}

\newcommand{\bapaSubvalid}{\sqsubseteq_\assmp}

\newcommand{\foassmp}{\Theta}
\newcommand{\IPFassmp}{\Delta}

\newcommand{\nodevar}{n}
\newcommand{\Intbapa}{\Int_B}
\newcommand{\Dombapa}{\Dom_B}

\newcommand{\thresholds}{T}
\newcommand{\thresholdsEx}{\hat{\thresholds}}
\newcommand{\guard}[1]{g_{\geq #1}}

\newcommand{\Voccore}{\Sigma_C}
\newcommand{\structcore}{\struct_C}
\newcommand{\Domcore}{\Dom_C}
\newcommand{\Intcore}{\Int_C}

\newcommand{\lang}{TIP\xspace}

\newcommand{\ITA}{\textsc{Aip}\xspace}
\newcommand{\ouralg}{\ITA}
\newcommand{\ouralgeager}{\textsc{Aip\textsubscript{Eager}}\xspace}
\newcommand{\ouralglazy}{\textsc{Aip\textsubscript{Lazy}}\xspace}
 
\begin{document}
\pagestyle{plain}

\setlength{\abovedisplayskip}{3pt}
\setlength{\belowdisplayskip}{3pt}
\setlength{\abovedisplayshortskip}{0pt}
\setlength{\belowdisplayshortskip}{0pt}
\title{Verification of Threshold-Based Distributed Algorithms by Decomposition to Decidable Logics
}

\author{
Idan Berkovits\inst{1}
\and
Marijana Lazi\' c\inst{2}\and
Giuliano Losa\inst{3}
\and
Oded Padon\inst{4}
\and
Sharon Shoham\inst{1}
 }
\institute{
Tel Aviv University, Israel
\and
TU Wien, Austria, TU Munich, Germany
\and
University of California, Los Angeles, USA
\and
Stanford University, USA
}
\maketitle              

\begin{abstract}
Verification of fault-tolerant distributed protocols is an
  immensely difficult task. Often, in these protocols,
\emph{thresholds} on set cardinalities are used both in the process code and in its correctness proof, e.g.,
  a process can perform an action only if it has received an
  acknowledgment from at least half of its peers.
Verification of threshold-based protocols is extremely challenging as it involves two kinds of reasoning:
first-order reasoning about the unbounded state of the protocol,
together with reasoning about sets and cardinalities.
In this work, we develop a new methodology for decomposing the verification task of such protocols into \emph{two} decidable logics: EPR and BAPA.
Our key insight is that such protocols use thresholds in a restricted way as a means
to obtain certain  properties of ``intersection'' between sets.
We define a language for expressing such properties, and present two translations: to EPR and BAPA.
The EPR translation allows verifying the protocol while assuming these properties, and the BAPA translation allows verifying the correctness of the properties.
We further develop an algorithm for automatically generating the properties needed for verifying a given protocol,
facilitating fully automated deductive verification.
Using this technique we have verified several challenging protocols,
  including Byzantine one-step consensus, hybrid reliable broadcast and
  fast Byzantine Paxos.
\end{abstract}

\section{Introduction}

Fault-tolerant distributed protocols play an important role in the avionic
and automotive industries, medical devices, cloud systems,
blockchains,~etc.  Their unexpected behavior might put human lives at
risk or cause a huge financial loss.  Therefore, their correctness is
of ultimate importance.

Ensuring correctness of distributed protocols is a notoriously difficult task, due to the unbounded number of processes and messages, as well as the non-deterministic behavior caused by the presence of faults, concurrency, and message delays.
In general, the problem of verifying such protocols is undecidable.
This imposes two directions for attacking the problem:
  (i) developing fully-automatic verification techniques for \emph{restricted} classes of protocols, or
  (ii) designing deductive techniques for a wide range of systems that \emph{require user assistance}.
Within the latter approach, recently emerging techniques~\cite{McMillanP18} leverage decidable logics that are supported by mature automated solvers to significantly reduce user effort, and increase verification productivity.
Such logics bring several key benefits: \begin{inparaenum}[(i)] \item their solvers usually enjoy stable performance, and
\item  whenever annotations provided by the user are incorrect, the automated solvers can provide a counterexample for the user to examine. \end{inparaenum}

Deductive verification based on decidable logic requires a logical formalism that satisfies two conflicting criteria: the formalism should be expressive enough to capture the protocol, its correctness properties, its inductive invariants, and ultimately its verification conditions. At the same time, the formalism should be decidable and have an effective automated tool for checking verification conditions.

In this paper we develop a methodology for deductive verification of \emph{threshold-based} distributed protocols using decidable logic, where we use \emph{decomposition} into two well-established decidable logics to settle the tension explained above.

In threshold-based protocols, a process may take different actions based on the number of processes from which it received certain messages.
This is often used to achieve fault-tolerance. For example, a process may take a certain step once it has received an acknowledgment from a strict majority of its peers, that is, from more than $\nP/2$ processes, where $\nP$ is the total number of processes.
Such expressions as $\nP/2$, are called \emph{thresholds}, and in general they can depend on additional parameters, such as the maximal number of crashed processes, or the maximal number of Byzantine processes.

Verification of such protocols requires two flavors of reasoning, as demonstrated by the following example.
Consider the Paxos~\cite{lamport_part-time_1998} protocol, in which each process proposes a value and all must agree on a common proposal.
The protocol tolerates up to~$\tP$ process crashes, and ensures that every two processes that decide agree on the decided value.
The protocol requires $\nP>2\tP$ processes, and each process must obtain confirmation messages from $\nP-\tP$ processes before making a decision. The protocol is correct due to, among others,
the fact that if $\nP>2\tP$ then any two sets of $\nP-\tP$ processes have a process in common.
To verify this protocol we need to express (i)~relationships between an unbounded number of processes and values, which typically requires quantification over uninterpreted domains (``every two processes''), and (ii)~properties of sets of certain cardinalities (``any two sets of $\nP-\tP$ processes intersect'').
Crucially, these two types of reasoning are intertwined, as the sets of processes for which we need to capture cardinalities may be defined by their relations with other state components (``messages from at least $\nP-\tP$ processes'').
While uninterpreted first-order logic (FOL) seems like the natural fit for the first type of reasoning, it is seemingly a poor fit for the second type, since it cannot express set cardinalities and the arithmetic used to define thresholds.
Typically, logics that combine both types of reasoning are either undecidable or not flexible enough to capture protocols as intricate as the ones we consider.

The approach we present relies on the observation
that threshold-based protocols and their correctness proofs use set cardinality thresholds in a restricted way as a means to obtain certain properties between sets, and that these properties can be expressed in FOL via a suitable encoding.
In the example above, the important property is that every two sets of cardinality at least $\nP-\tP$ have a non-empty intersection.
This property can be encoded in FOL by modeling sets of cardinality at least $\nP-\tP$ using an uninterpreted sort along with
a \iflong corresponding \fi membership relation between this sort and the sort for processes.
However, the validity of the \iflong intersection \fi property under the assumption that $\nP > 2\tP$ cannot be verified in FOL.

The key idea of this paper is, hence, to decompose the verification problem of threshold-based protocols into the following problems:
\begin{inparaenum}[(i)]
    \item Checking protocol correctness  assuming certain intersection properties, which can be reduced to verification conditions expressed in the Effectively Propositional (EPR) fragment of FOL~\cite{lewisComplexityResultsClasses1980,piskac_deciding_2010}.
    \item Checking that sets with cardinalities adhering to the thresholds satisfy the intersection properties (under the protocol assumptions), which can be reduced to validity checks in quantifier-free Boolean Algebra with Presburger Arithmetic (BAPA)~\cite{KuncakNR05}.
\end{inparaenum}
Both BAPA and EPR are decidable logics, and are supported by mature \iflong automated \fi solvers.

A crucial step in employing this decomposition is finding suitable intersection properties that are strong enough to imply the protocol's correctness (i.e., imply the FOL verification conditions), and are also implied by the precise definitions of the thresholds and the protocol's assumptions.
Thus, these intersection properties can be viewed as \emph{interpolants} between the FOL verification conditions and the thresholds in the context of the protocol's assumptions. We present fully automated procedures to find such intersection property interpolants, either eagerly or lazily.

The main contributions of this paper are:
\begin{enumerate}[noitemsep,topsep=0pt]
  \item We define a threshold intersection property (\lang) language for expressing properties of sets whose cardinalities adhere to certain thresholds; \lang is expressive enough to capture the properties required to prove the correctness of challenging threshold-based protocols.

\item We develop two encodings of \lang, one in BAPA, and another in EPR.
These encodings facilitate decomposition of protocol verification into decidable EPR and (quantifier-free) BAPA queries.
 \item We show that there are only finitely many \lang formulas (up to equivalence) that are valid for any given protocol.
Moreover, we present an effective algorithm for computing all \lang formulas valid for a given protocol, as well as an algorithm for lazily finding a set of \lang formulas that suffice to prove a given protocol.
\item Put together, we obtain an effective deductive verification approach for threshold-based protocols:
  the user models the protocol and its inductive invariants in EPR using a suitable encoding of thresholds, and defines the thresholds and the protocol's assumptions using arithmetic; verification is carried out automatically via decomposition to well-established decidable logics.
  \item We implement the approach, leveraging mature existing solvers (Z3 and CVC4), and evaluate it by verifying several challenging threshold-based protocols with sophisticated thresholds and assumptions. Our evaluation shows the effectiveness and flexibility of our approach in modeling and verifying complex protocols, including the feasibility of automatically inferring threshold intersection properties.

\end{enumerate}

 \section{Preliminaries}
\label{sec:background}

\iflong
In this section we present the necessary background on the formalization of transition systems using FOL, as well as on the decidable logics used in our work:~EPR and~BAPA.
\fi

\paragraph{\bf Transition Systems in FOL}

We model distributed protocols as transition systems expressed in many-sorted FOL. A state of the system is a first-order (FO) structure $\struct = (\Dom, \Int)$
over a vocabulary $\Sigma$ that consists of sorted constant, function and relation symbols, s.t. $\struct$ satisfies a finite set of \emph{axioms} $\foassmp$ in the form of closed formulas over $\Sigma$. $\Dom$ is the \emph{domain} of $\struct$
mapping each sort to a set of objects (elements),
and $\Int$ is the \emph{interpretation function}.
\iflong
A FO \emph{transition system} is a tuple $(\Sigma, \foassmp, \init, \TR)$, where $\Sigma$ and $\foassmp$ are as above, $\init$ is the \emph{initial condition} given by a closed formula over $\Sigma$, and $\TR$ is the
\emph{transition relation} given by a closed formula over $\Sigma
\uplus \Sigma'$ where $\Sigma$ describes the source state of
the transition and $\Sigma' = \{a' \mid a \in \Sigma\}$
describes the target state. We require that $\TR$ does not modify any symbol that appears in $\foassmp$. The set of initial states and the set of
transitions of the system consist of the states, respectively, pairs
of states, that satisfy $\init$, respectively, $\TR$.
\fi
The set of reachable states is defined as usual.
In practice, we define FO transition systems using a modeling language with a convenient syntax~\cite{McMillanP18}.

\paragraph{Properties and inductive invariants.}
A \emph{safety property} is expressed by a closed FO formula $\safe$
over $\Sigma$. The system is \emph{safe} if all of its reachable
states satisfy $\safe$.
A closed FO formula $\Inv$ over $\Sigma$ is an
\emph{inductive invariant} for a transition system $(\Sigma, \foassmp, \init, \TR)$ and property $\safe$ if
the following formulas, called the \emph{verification conditions},  are valid (equivalently, their negations are unsatisfiable):
(i) $\foassmp \to (\init \to \Inv)$, (ii) $\foassmp \to (\Inv \wedge \TR \to \Inv')$ and (iii)~$\foassmp \to (\Inv \to \safe)$, where
$\Inv'$ results from substituting every symbol in $\Inv$ by its primed
version. \iflong Requirements (i) and (ii) ensure that $\Inv$
represents a superset of the reachable states, hence together with (iii) safety follows. \fi
We also use inductive invariants to verify arbitrary first-order LTL formulas via the reduction of~\cite{PadonHLPSS18,PadonHMPSS18}.

\paragraph{\bf Effectively Propositional Logic (EPR)}

The effectively-propositional (EPR) fragment of FOL
is restricted to
formulas without function symbols and with a quantifier prefix $\exists^* \forall^*$ in prenex
normal form.  Satisfiability of EPR formulas is
decidable~\cite{lewisComplexityResultsClasses1980}.  Moreover, EPR formulas
enjoy the \emph{finite model property}, \iflong meaning that a satisfiable
formula is guaranteed to have a finite model.
\fi
\iflong
The size of this model
is bounded by the total number of existential quantifiers and
constants in the formula.
\fi
\iflong
While EPR does not allow any function symbols nor quantifier alternation except $\exists^* \forall^*$,
we consider
\fi
a straightforward extension of EPR that maintains these properties and is supported by \iflong mature \fi solvers such as Z3~\cite{de_moura_z3:_2008}. The extension allows function symbols and quantifier alternations
as long as the formula's \emph{quantifier alternation
graph}, denoted $\stgraph{\varphi}$, is acyclic.
\label{sec:qa-graph}
For $\varphi$ in negation normal form, $\stgraph{\varphi}$
is a directed graph where the set of vertices is the set of sorts and the set of edges
is defined as follows:
\iflong
\begin{itemize}[noitemsep,topsep=0pt]
\item {\bf Function edges:} let $\func$ be a function in $\varphi$ from sorts $\srt_1,\ldots,\srt_k$ to sort $\srt$. Then there is a $\forall\exists$ edge from $\srt_i$ to $\srt$ for every $1 \leq i \leq k$.
\item {\bf Quantifier edges:} let $\exists x: \srt$ be an existential quantifier that resides in the scope of the universal quantifiers $\forall x_1: \srt_1, \ldots, \forall x_k: \srt_k$ in $\varphi$.
Then there is a $\forall\exists$ edge from $\srt_i$ to $\srt$ for every $1 \leq i \leq k$.
\end{itemize}
\fi
The quantifier alternation graph is extended to sets of formulas as expected.

\paragraph{\bf Boolean Algebra with Presburger Arithmetic (BAPA)} \label{sec:bapa}

Boolean Algebra with Presburger Arithmetic (BAPA)~\cite{KuncakNR05} is a FO theory defined over two sorts:
\iflong $\sbapaint$, representing integers, and $\sbapaset$, representing subsets of a finite universe.
\fi
The language is defined \iflong by the following grammar\fi:
\iflong
\begin{align*}
\Fgrammar ::= & \Bgrammar_1 = \Bgrammar_2 \mid \Lgrammar_1 = \Lgrammar_2 \mid \Lgrammar_1 < \Lgrammar_2  \mid  \Fgrammar_1 \wedge  \Fgrammar_2 \mid \Fgrammar_1 \vee \Fgrammar_2 \mid \neg \Fgrammar \mid \exists x.\Fgrammar \mid \forall x.\Fgrammar \mid \exists u.\Fgrammar \mid \forall u.\Fgrammar \\
\Bgrammar ::= & x \mid \emptyset \mid \cSet \mid \Bgrammar_1 \cup \Bgrammar_2 \mid  \Bgrammar_1 \cap \Bgrammar_2 \mid \Bgrammar^c \\ \Lgrammar ::= & u \mid  \Kgrammar \mid \nP \mid \cInt \mid \Lgrammar_1 + \Lgrammar_2 \mid \Kgrammar \cdot \Lgrammar \mid |\Bgrammar| \\
\Kgrammar ::= & \ldots -2 \mid -1 \mid  0 \mid 1 \mid 2 \ldots
\end{align*}
\fi
where $\Lgrammar$ defines linear integer terms, where $u$ denotes an integer variable,
$k \in \Kgrammar$ defines an (interpreted) integer constant symbol $\ldots, -2,-1,0,1,2 \ldots $,
$\nP$ is an integer constant symbol that represents the size of the finite set universe, $\cInt$ is an uninterpreted integer constant symbol (as opposed to the constant symbols from $\Kgrammar$), and $|\Bterm|$ denotes set cardinality; $\Bgrammar$ defines set terms, where $x$ denotes a set variable, $\emptyset$ is a (interpreted) set constant symbol that represents the empty set, and $\cSet$ is an uninterpreted set constant symbol;
and $\Fgrammar$ defines the set of BAPA formulas, where $\ell_1 = \ell_2$ and $\ell_1 < \ell_2$ are atomic arithmetic formulas and $\Bterm_1 = \Bterm_2$ is an atomic set formula.
(Other set constraints such as $\Bterm_1 \subseteq \Bterm_2$ can be encoded in the usual way).
\iflong Note that BAPA formulas are closed under Boolean operations and allow quantification over integer variables ($u$), and over set variables ($x$). \fi
In the sequel, we also allow arithmetic terms of the form $\frac{\ell}{k}$ where $k \in K$ is a positive integer and $\ell \in \Lgrammar$, as any formula that contains such terms can be translated to an equivalent BAPA formula by multiplying by $k$.

\iflong
A BAPA structure $\structbapa = (\Dom, \Int)$ consists of a domain $\Dom$ mapping sort $\sbapaint$ to the set of all integers and mapping the sort $\sbapaset$ to the set of all subsets of a finite universe $U$, called the \emph{universal set}, and an interpretation function $\Int$ of the symbols for integer and set operations. The semantics of terms and formulas is as expected.
The interpretation of the complement operation is defined with respect to $U$. In particular, this means that $\Int(\emptyset^c) = U$.
The integer constant $\nP$ is interpreted to the size of $U$, i.e. $\Int(\nP) = \card{U}$.
The uninterpreted constant symbols may be interpreted arbitrarily.
\fi

Both validity and satisfiability of BAPA formulas (with arbitrary
quantification) are decidable~\cite{KuncakNR05}, and the quantifier-free fragment is supported by CVC4~\cite{BansalNewDecisionProcedure2016}\iflong, a mature SMT solver\fi.

 \section{First-Order Modeling of Threshold-Based Protocols}
\label{sec:decomp}

Next we explain our modeling of threshold-based
protocols as transition systems in~FOL (Note that FOL cannot
directly express set cardinality
constraints).
The idea is to capture each threshold by a designated
sort, such that elements of this sort represent sets of
nodes that satisfy the threshold.  Elements of the threshold
sort are then used instead of the actual threshold in the description
of the protocol and in the verification conditions.
For verification
to succeed, some properties of the sets satisfying the cardinality threshold must be
captured in FOL. This is done by introducing additional assumptions
(formally, axioms of the transition system) expressed in FOL, as discussed in \Cref{sec:language}.

\paragraph{Running Example.}\label{sec:exampleBosco}

\begin{figure}[t]
\centering
\lstset{ basicstyle=\scriptsize,columns=flexible,
    keepspaces=true,
    numbers=left,
xleftmargin=2em,
    numberstyle=\tiny,
    emph={
if, then, else, wait, until, broadcast, call
},
    emphstyle={\bfseries},
    mathescape=true,
    frame=none
}
\begin{tabular}{cc}
\begin{minipage}{0.45\textwidth}
\begin{tabular}{cc}
\begin{lstlisting}
Input: $v_p$
broadcast $v_p$ to all processes
wait until $n-t$ messages have been received

if there exists $v$ s.t. more than $\frac{n+3t}{2}$  $\label{line:bosco-pseudo-enough}$
    messages contain value $v$ then
  DECIDE($v$)
if there exists exactly one $v$ s.t. more than
      $\frac{n-t}{2}$ messages contain value $v$ then $\label{line:bosco-pseudo-unique}$
  $v_p$ := $v$
call underlying-consensus($v_p$)
\end{lstlisting} \hspace{0.2cm}
\end{tabular}
\end{minipage}
&
\begin{minipage}{0.5\textwidth}
\begin{lstlisting}
sort $\snode$, $\svalue$, $\squorumof{\nP - \tP}$, $\squorumof{\frac{\nP + 3 \tP + 1}{2}}$, $\squorumof{\frac{\nP - \tP + 1}{2}}$
$\cdots$
assume $\exists q : \squorumof{\nP - \tP}. \; \forall m : \snode.\; \rmember(m,q) \to$
                $\exists u : \svalue. \; \rrcvmsg(n,m,u)$
if $\exists v:\svalue,\, q : \squorumof{\frac{\nP + 3 \tP + 1}{2}}. \; \forall m : \snode.\; $ $\label{line:bosco-rml-enough}$
          $\rmember(m,q) \to \rrcvmsg(n,m,v)$ then
    $\rdecision(n,v)$ := true
if $\exists! v:\svalue. \; \exists q : \squorumof{\frac{\nP - \tP + 1}{2}}. \; \forall m : \snode.\; $
          $\rmember(m,q) \to \rrcvmsg(n,m,v)$ then $\label{line:bosco-rml-unique}$
    $v_p$ := $v$
$\rundcons(n,v_p)$ := true
\end{lstlisting}
\end{minipage}
\end{tabular}
\caption{\label{fig:bosco-pseudo} \label{fig:bosco-rml}
Bosco: a one-step asynchronous Byzantine consensus algorithm~\cite{SongR08}, and
an excerpt RML (relational modeling language) code of the main transition. Note that we overload the $\rmember$ relation for all threshold sorts.
The formula $\exists! x.\,\varphi(x)$ is a shorthand for
\iflong
$(\exists x.\, \varphi(x)) \land (\forall x,y.\,\varphi(x) \land \varphi(y) \to x=y)$.
\fi
}
\end{figure}

We illustrate our approach
using the example of Bosco---an
asynchronous Byzantine fault-tolerant (BFT) consensus algorithm~\cite{SongR08}.
Its modeling in first-order logic using our technique appears alongside an
informal pseudo-code in~\Cref{fig:bosco-rml}.

In the BFT consensus problem, each node proposes a value and correct nodes must
decide on a unique proposal. BFT consensus algorithms typically require at
least two communication rounds to reach a decision.  In Bosco, nodes execute a
preliminary communication step which, under favorable conditions, reaches an
early decision, and then call an underlying BFT consensus algorithm to ensure
reaching a decision
even if these conditions are not met.
Bosco is safe when $\nP > 3\tP$; it guarantees that a preliminary decision will be reached if all nodes are non-faulty and propose the same value when~$\nP > 5\tP$ (weakly one-step condition), and even if some nodes are faulty, as long as all non-faulty nodes propose the same value, when~$\nP > 7\tP$ (strongly one-step condition).

Bosco achieves consensus by ensuring that (a)~no two correct nodes decide
differently in the preliminary step, and (b)~if a correct node decides
value $v$ in the preliminary step then every correct process calls the
underlying BFT consensus algorithm with proposal $v$.  Property (a) is ensured
by the fact that a node decides in the preliminary step only if more than
$\frac{\nP+3\tP}{2}$ nodes proposed the same value. When $\nP>3\tP$, two sets of cardinality
greater than $\frac{\nP+3\tP}{2}$ have at least one non-faulty node in common, and
therefore no two different values can be proposed by more than $\frac{\nP+3\tP}{2}$ nodes.  Similarly, we can derive property (b) from the fact that a set of more
than $\frac{\nP+3\tP}{2}$ nodes and a set of $\nP-\tP$ nodes intersect in $\frac{\nP+\tP}{2}$ nodes,
which, after removing $t$ nodes which may be faulty, still leaves us with
more than $\frac{\nP-\tP}{2}$ nodes, satisfying the condition in line 9.

\subsection{Threshold-based protocols}

\paragraph{Parameters and resilience conditions.}
We consider protocols whose definitions depend on a set of \emph{parameters}, $\Param$, divided into \emph{integer parameters}, $\ParamInt$, and \emph{set parameters}, $\ParamSet$.
$\ParamInt$ always includes $\nP$, \iflong used to represent \fi the total number of nodes (assumed to be finite).
Protocol correctness is ensured
under a set of assumptions $\assmp$ called \emph{resilience conditions}, formulated as BAPA formulas over $\Param$ (this means that
all the uninterpreted constants appearing in $\assmp$ are from $\Param$).
In Bosco, $\ParamInt = \{\nP, \tP\}$, where $\tP$ is the maximal number of Byzantine failures tolerated by the algorithm,
and $\ParamSet = \{ \fP \}$, where $\fP$ is the set of Byzantine nodes;
$\assmp = \{ \nP \geq 3 \tP + 1, \card{\fP} \leq \tP \}$.

\paragraph{Threshold conditions.}
Both the description of the protocol and the inductive invariant may include conditions that require the size of some set of nodes to be ``at least $t$'', ``at most $t$'', and so on, where the threshold $t$ is of the form $t=\frac{\ell}{k}$, where $k$ is a positive integer, and $\ell$ is a ground BAPA integer term over $\Param$
\iflong (we restrict $k$ to ensure that the guard can be translated to BAPA). Comparing sizes of two sets is not allowed -- we observe that it is not needed for threshold-based protocols.
\fi
We denote the set of thresholds by $\thresholds$.
For example, in Bosco, $\thresholds = \{ \nP - \tP, \frac{\nP + 3 \tP + 1}{2}, \frac{\nP - \tP + 1}{2} \}$.

\iflong
Without loss of generality, we assume that all conditions on set cardinalities are of the form ``at least $t$''.
This is because every condition can be written in this form, possibly by introducing new threshold expressions and complementing the set which the condition refers to, according to the following rules:
\fi
\[
\card{X} > \frac{\ell}{k} \equiv \card{X} \geq \frac{\ell+1}{k} \qquad \card{X} \leq \frac{\ell}{k} \equiv \card{X^c} \geq \frac{k \cdot \nP- \ell}{k} \qquad \card{X} < \frac{\ell}{k}\equiv \card{X} \leq \frac{\ell-1}{k}
\]

\subsection{Modeling in FOL} \label{subsec:decomp:fol-modeling}

\paragraph{FO vocabulary for modeling threshold-based protocols.}

We describe the protocol's states (e.g., pending messages, votes, etc.) using a core FO vocabulary $\Voccore$ that includes sort $\snode$ and additional sorts and symbols.
Parameters $\Param$ are \emph{not} part of the FO vocabulary used to model the protocol.
Also, we do not model set cardinality directly.
Instead, we encode the cardinality thresholds in FOL by defining a FO vocabulary $\Sigma_{\thresholds}^{\Param}$: \begin{compactitem}
\item For every threshold  $t$ we introduce a \emph{threshold sort} $\squorumof{t}$ with the intended meaning that elements of this sort are sets of nodes whose size is at least $t$.
\item \iflong As the threshold sorts represent \emph{sets}, each sort $\squorumof{t}$ is equipped with a binary relation symbol $\rmember_{t}$  between sort $\snode$ and sort $\squorumof{t}$ that captures the membership relation of a node (first argument) in a set (second argument).
\fi
\item For each set parameter $\cSet \in \ParamSet$ we introduce a unary relation symbol $\rmember_{\cSet}$ over sort $\snode$ that captures membership of a node in the set $\cSet$.
\end{compactitem}
We then model the protocol as a transition system $(\Sigma, \foassmp, \init, \TR)$ where $\Sigma=\Voccore\uplus \Sigma_{\thresholds}^{\Param}$.

We are interested only in states (FO structures over $\Sigma$) where the interpretation of the threshold sorts and membership relations is according to their intended meaning in a corresponding BAPA structure.
\iflong
The intended meaning is captured by a BAPA structure (over $\Param$) that satisfies the resilience condition, and whose universal set coincides with the set of nodes. \fi
Formally, these are $\thresholds$-extensions, defined as follows:

\begin{definition}
\label{def:extension} \label{def:faithful}
We say that a FO structure $\structcore = (\Domcore, \Intcore)$ over $\Voccore$ and a BAPA structure $\structbapa = (\Dombapa, \Intbapa)$ over $\Param$ are \emph{compatible} if $\Dombapa(\sbapaset) = \powerset{\Domcore(\snode)}$, where $\mathcal{P}$ is the powerset operator.
For such compatible structures and a set of thresholds $\thresholds$ over $\Param$,
the \emph{$\thresholds$-extension} of $\structcore$ by $\structbapa$ is the structure $\struct = (\Dom, \Int)$ over  $\Sigma$
defined as follows:
\indent\begin{tabular}{ll}
$\Dom(\bsort) =  \Domcore(\bsort)  \mbox { for every sort $\bsort$ in } \Voccore\qquad \qquad$
& $\Int(\letter) =  \Intcore(\letter) \mbox { for every $\letter$ in } \Voccore$ \\
$\Dom(\squorumof{t}) = \{A \subseteq \Domcore(\snode) \mid \card{A} \geq \Intbapa(t) \}$
& $\Int(\rmember_{\cSet}) = \Intbapa(\cSet)$ \\
\multicolumn{2}{l}{$\Int(\rmember_{t}) = \{ (e,A) \mid e \in \Domcore(\snode), A \in \Dom(\squorumof{t}), e \in A  \}$ }\\
\end{tabular}
\end{definition}

\noindent Note that for the $\thresholds$-extension $\struct$ to be well defined as a FO structure, we must have that $\Dom(\squorumof{t}) \neq \emptyset$ for every threshold $t \in T$.
This means that a $\thresholds$-extension by $\structbapa$ only exists if $\{A \subseteq \Dom(\snode) \mid \card{A} \geq \Intbapa(t) \} \neq \emptyset$\iflong, i.e., there exists at least one set that satisfies each threshold in $\thresholds$\fi.
This is ensured by the following condition:

\begin{definition}[Feasibility]
$\thresholds$ is \emph{$\assmp$-feasible} if $\assmp \models t \leq \nP$ for every $t \in T$.
\end{definition}

\paragraph{Expressing threshold constraints.}

Cardinality constraints can be expressed in FOL over the vocabulary $\Sigma=\Voccore\uplus \Sigma_{\thresholds}^{\Param}$ using quantification. To express \iflong the condition \fi that $\card{\{ \nodevar : \snode  \mid \varphi(\nodevar,\bar{u}) \}} \geq t$, i.e., that there are at least $t$ nodes that satisfy the FO formula $\varphi$ over $\Voccore$ (where $\bar{u}$ are \iflong additional \fi free variables in $\varphi$),
we use the following first-order formula over $\Sigma$:
$
\exists q : \squorumof{t}. \; \forall \nodevar : \snode. \; \rmember_t(\nodevar,q) \to \varphi(\nodevar,\bar{u})
$\iflong,  where $\squorumof{t}$ is the threshold sort assigned to $t$\fi.
Similarly, to express the property that a node is a member of a set parameter $\cSet$ (e.g., to check if $\nodevar \in \fP$, i.e., a node is faulty) we use the FO formula $\rmember_{\cSet}(\nodevar)$.
For example, in \Cref{fig:bosco-rml},  \cref{line:bosco-rml-enough} (right) uses the FO modeling to express the condition in \cref{line:bosco-pseudo-enough} (left).
This modeling is sound in the following sense:

\begin{lemma}[Soundness]
\label{lem:soundness-extension}
    Let $\structcore = (\Domcore, \Intcore)$ be a FO structure over $\Voccore$, $\structbapa = (\Dombapa, \Intbapa)$ a compatible BAPA structure over $\Param$
    s.t. $\structbapa \models \assmp$ and $\thresholds$ a $\assmp$-feasible set of thresholds over $\Param$. Then
    there exists a (unique) $\thresholds$-extension $\struct$ of $\structcore$ by $\structbapa$. Further:
\begin{enumerate}[topsep=0pt]
        \item For every $\cSet \in \ParamSet$ and FO valuation $\Val$: $\struct, \Val \models \rmember_{\cSet}(\nodevar)$ iff $\Val(\nodevar) \in \Intbapa(\cSet)$,
        \item For every $t \in \thresholds$, formula $\varphi$, and FO valuation $\Val$:  $\struct, \Val \models \exists q : \squorumof{t}. \; \forall \nodevar : \snode. \; \rmember_t(\nodevar,q) \to \varphi(\nodevar,\bar{u})$ iff          $\card{\{ e \in \Dom(\snode) \mid \structcore,\Val[\nodevar \mapsto e] \models \varphi(\nodevar,\bar{u}) \}} \geq \Intbapa(t)$.
    \end{enumerate}
\end{lemma}

\begin{definition}
A first-order structure $\struct$ over $\Sigma$ is \emph{threshold-faithful} if it is a $\thresholds$-extension of some $\structcore$ by some $\structbapa \models \assmp$ (as in \Cref{lem:soundness-extension}).
\end{definition}

\paragraph{Incompleteness.}
\Cref{lem:soundness-extension} ensures that the FO
modeling can be soundly used to verify the protocol.
It also ensures that the modeling is precise on
threshold-faithful  structures (\Cref{def:extension}).  Yet, the
FO transition system is not restricted to such
states, 
\iflong
making it an \emph{abstraction} of the actual protocol.
\fi
 To have any hope to verify the protocol, we must capture
\emph{some} of the intended meaning of the threshold sorts and
relations. This is obtained by adding FO axioms to
the FO transition system.  Soundness is maintained as long as
the axioms 
hold in all threshold-faithful structures.
We note that the set of \emph{all} \iflong FO formulas over $\Sigma$ that hold in all threshold-faithful
structures \fi is not recursively
enumerable\iflong\footnote{In a threshold-faithful structure the set of nodes
  must be finite, so validity of a general formula in all
  threshold-faithful structures can easily encode validity of FOL
  over finite structures, which has no complete proof system.}, which is where the essential incompleteness of our approach lies.\fi

 \section{Decomposition via Threshold Intersection Properties} \label{sec:language}

In this section,
we identify a set of properties we call \emph{threshold intersection properties}. When captured via FO axioms,
these properties
suffice for verifying
many threshold-based protocols (all the ones we considered).
Importantly, these are properties of sets adhering to the thresholds that do not involve the protocol state.
As a result, they can be expressed both in FOL and in BAPA.
This allows us
to decompose the
verification task into:
\begin{inparaenum}[(i)]
\item checking that certain threshold properties are valid in all threshold-faithful structures by checking that they are implied by $\assmp$
(carried out using quantifier free BAPA), and
\item checking that the verification conditions of the FO transition-system with the same threshold properties taken as axioms are valid (carried out in first-order logic, and in EPR if quantifier alternations are acyclic).
\end{inparaenum}

\subsection{Threshold Intersection Property Language}

Threshold properties are expressed in the \emph{threshold intersection property language} (\lang).
\lang is essentially a subset of BAPA, specialized to have the properties listed above. 

\paragraph{Syntax.}
\iflong
We define \lang as follows, where $t \in \thresholds$ is a threshold (of the form  $\frac{\ell}{k}$) and $\cSet \in \ParamSet$:
\fi
\begin{align*}
F & :: = B \neq \emptyset \mid B^c = \emptyset \mid \guard{t}(B) \mid F_1 \land F_2 \mid \forall x:\guard{t}. F \\
B & ::= \cSet \mid \cSet^c \mid x \mid x^c \mid \emptyset \mid \emptyset^c \mid  B_1 \cap B_2
\end{align*}
\iflong
\lang restricts the use of set cardinality to \emph{threshold guards} $\guard{t}(\Bterm)$ with the meaning $\card{\Bterm} \geq t$.
No other arithmetic atomic formulas \iflong ($\ell_1 = \ell_2$ or $\ell_1 < \ell_2$)  \fi are allowed.
\lang only allows guarded quantifiers over set variables,
that is, quantification is restricted to sets of a certain cardinality.
We exclude negation, disjunction and existential quantification in formulas.
We restrict comparison atomic formulas to $\Bterm \neq \emptyset$ and $\Bterm^c = \emptyset$\iflong, which correspond to asserting that the cardinality of the set represented by $\Bterm$ is \emph{at least} $1$, respectively $\nP$\fi.
Furthermore, we forbid set union and restrict complementation to atomic set terms.
\fi
We refer to such formulas as \emph{intersection properties}  since they express properties of intersections of (atomic) sets.

\begin{example}
In Bosco, the following property captures the fact that the intersection of a set of at least
$\nP-\tP$ nodes and a set of more
than $\frac{\nP+3\tP}{2}$ nodes consists of at least $\frac{\nP-\tP}{2}$ non-faulty nodes. This is needed for establishing correctness of the protocol.
\begin{equation*}
\forall x:\guard{\nP-\tP}. \, \forall y:\guard{\frac{\nP + 3\tP + 1}{2}}. \; \guard{\frac{\nP - \tP + 1}{2}}(x \cap y \cap \fP^c)
\end{equation*}
\end{example}

\paragraph{Semantics.}

As \lang is essentially a subset of BAPA, we define its semantics by translating its formulas to BAPA,
where most constructs directly correspond to BAPA constructs, and guards are translated to cardinality constraints:
\iflong
Set terms (derived from $B$) are also set terms of BAPA, and most set formula constructs map to constructs of BAPA directly.
The only constructs that are not in BAPA are those involving guards, which correspond to a special case of cardinality constraints. We therefore define the following translation $\ToBAPA$:
\fi
\begin{align*}
&\ToBAPA(g_{\geq \frac{\ell}{k}}(\Bterm))  \eqdef k \cdot \card{\Bterm} \geq \ell  \qquad \ToBAPA(\forall x:g.\ \varphi)  \eqdef \forall x.\ \neg \ToBAPA(g(x)) \vee \ToBAPA(\varphi)
\end{align*}

The notions of structures, satisfaction, equivalence, validity, satisfiability, etc. are inherited from BAPA. In particular, given a set of BAPA resilience conditions $\assmp$ over the parameters $\Param$, we say that a \lang formula $\varphi$ is $\assmp$-valid, denoted $\assmp \models \varphi$, if $\assmp \models \ToBAPA(\varphi)$.

If $\assmp$ is quantifier-free (which is the typical case), $\assmp$-validity of \lang formulas can be checked via validity checks of quantifier-free BAPA formulas, supported by mature solvers.
Note that $\assmp$-validity of a formula of the form $\forall x:\guard{t_1}.\ \card{x \cap \Bterm} \geq t_2$ is equivalent to $\Gamma \models \forall u.\ u \geq t_1 \rightarrow u + \card{\Bterm} -n \geq t_2$, allowing to replace quantification over sets by quantification over integers, thus improving performance of existing solvers.

\subsection{Translation to FOL}
\label{sec:translation-to-fol}

To verify threshold-based protocols, we
translate \lang formulas to FO axioms, using the threshold sorts and relations.
To translate $\guard{t}(\Bterm)$, we follow the principle in (\Cref{subsec:decomp:fol-modeling}):

\noindent
\begin{tabular}{rclcrcl}
$ \ToFOL(\neg \varphi) $& = &$\neg \ToFOL(\varphi)$ & $\qquad$ &
$\ToFOL(\nodevar \in \Bterm^c)$ &=& $ \neg \ToFOL(\nodevar \in \Bterm) $
\\
$\ToFOL(\varphi_1 \wedge \varphi_2) $& = &$\ToFOL(\varphi_1) \wedge \ToFOL(\varphi_2)$&&
$\ToFOL(\nodevar \in \emptyset)$& =  &$\false$
\\
$\ToFOL(\forall\, x:g. \,\varphi)$ & = &$\forall\, x:\squorumof{g}\,. \ToFOL(\varphi)$&&
$\ToFOL(\nodevar \in \cSet)$& =& $\rmember_\cSet(\nodevar)$
\\
$\ToFOL(\nodevar \in \Bterm_1 \cap \Bterm_2) $&=& $\ToFOL(\nodevar \in \Bterm_1) \land \ToFOL(\nodevar \in \Bterm_2) $&&
$\ToFOL(\nodevar \in x) $&=&$ \rmember_t(\nodevar,x)$
\\
$\ToFOL(\Bterm \neq \emptyset)$ &=& $\exists \nodevar : \snode. \, \ToFOL(\nodevar \in \Bterm)$ &&
\multicolumn{3}{r}{where $x$ is guarded by $t$}
\\
$\ToFOL(\Bterm^c = \emptyset)$ &=& \multicolumn{5}{l}{$\forall \nodevar : \snode. \, \ToFOL(\nodevar \in \Bterm)$}
\\
$\ToFOL(\guard{t}(\Bterm)) $&=& \multicolumn{5}{l}{$\exists x:\squorumof{t}.\, \forall \nodevar:\snode. \, \rmember_t(\nodevar,x) \to \ToFOL(\nodevar \in \Bterm) $}
\\
\end{tabular}\\
We lift $\ToFOL$ to sets of formulas: $\ToFOL(\IPFassmp) = \{\ToFOL(\varphi) \mid \varphi \in \IPFassmp\}$.

\iflong\paragraph{Soundness of the translation to FOL.}\fi
Next, we state the soundness of the translation, which intuitively means that $\ToFOL(\varphi)$ is ``equivalent'' to $\varphi$ over threshold-faithful FO structures (\Cref{def:faithful}). This justifies
adding $\ToFOL(\varphi)$ as a FO axiom whenever $\varphi$ is $\assmp$-valid.

\begin{theorem}[Translation soundness]
\label{thm:soundness-translation}
Let $\structcore = (\Domcore, \Intcore)$ be a first-order structure over $\Voccore$, $\structbapa = (\Dombapa, \Intbapa)$ a compatible BAPA structure over $\Param$, and
$\struct$ the $\thresholds$-extension of $\structcore$ by $\structbapa$.
Then for every closed \lang formula $\varphi$, we have
$
\structbapa \models \varphi \Leftrightarrow
\struct \models \ToFOL(\varphi).
$ \end{theorem}

\begin{corollary}
For every closed \lang formula $\varphi$ such that $\assmp \models \varphi$, we have that $\ToFOL(\varphi)$ is satisfied by every threshold-faithful first-order structure.
\end{corollary}
\iflong
This justifies using the translation to FOL in order to generate first-order axioms from \lang formulas $\varphi$ that are entailed by the resilience conditions.
Namely, if $\assmp \models \varphi$, then  $\ToFOL(\varphi)$ may be safely added as a first-order axiom.
\fi

 \section{Automatically Inferring Threshold Intersection Properties}
\label{sec:axiom-synt}

\iflong
To apply the approach described in \Cref{sec:language,sec:decomp}
for verifying threshold-based protocols\fi,
it is crucial to find suitable
\iflong threshold properties for a given protocol.\fi That is, given the
resilience conditions $\assmp$ and a FO transition system
modeling the protocol, we need to find a set $\IPFassmp$ of \lang
formulas such that
\begin{inparaenum}[(i)]
    \item \label{pro:BapaValidity} $\assmp \models \varphi$ for every $\varphi \in \IPFassmp$, and
    \item \label{pro:FolValidity} the VCs of the transition system with the axioms
$\ToFOL(\IPFassmp)$ are valid.
\end{inparaenum}

In this section, we address the problem of automatically inferring such a set
$\IPFassmp$.
In particular, we prove that for any protocol that satisfies a natural condition, there are finitely many
$\assmp$-valid \lang formulas (up to equivalence), enabling a complete automatic inference
algorithm. Furthermore, we show that (under certain reasonable conditions formalized in this section), the FO axioms resulting from the inferred \lang properties have an \emph{acyclic} quantifier alternation graph, facilitating protocol verification in EPR.

\paragraph{Notation.}
For the rest of this section, we fix a set $\Param$ of parameters, a set $\assmp$ of resilience conditions over $\Param$, and a set $\thresholds$ of thresholds. Note that $\Bterm \neq \emptyset \equiv \guard{1}(\Bterm)$ and $\Bterm^c = \emptyset \equiv \guard{\nP}(\Bterm)$.
Therefore, for uniformity of the presentation, given a set $\thresholds$ of thresholds, we define $\thresholdsEx \eqdef \thresholds \cup \{1, \nP\}$ and replace atomic formulas of the form $\Bterm \neq \emptyset$ and $\Bterm^c = \emptyset$ by the corresponding guard formulas. As such, the only atomic formulas are of the form $\guard{t}(\Bterm)$ where $t \in \thresholdsEx$. Note that guards in quantifiers are still restricted to $\guard{t}$ where $t \in \thresholds$.
Given a set $\ParamSet$, we also denote $\hat{\ParamSet} = \ParamSet \cup \{\cSet^c \mid \cSet \in \ParamSet\}$.

\subsection{Finding Consequences in the Threshold Intersection Property Language}

In this section, we present \ITA -- an  algorithm for inferring all $\assmp$-valid \lang formulas.
A na\"ive (non-terminating) algorithm would iteratively check $\assmp$-validity of every \lang formula. 
\iflong 
Instead, \ouralg prunes the search space relying on the following condition, which essentially states that no guard is ``equivalent'' to $\guard{0}$ with respect to $\assmp$.
\fi

\begin{definition}\label{def:nondegenrateGuard}
$\thresholds$ is \emph{$\assmp$-non-degenerate} if for every $t \in \thresholds$ it holds that $\assmp \not\models t \leq 0$.
\end{definition}
If $\assmp \models t \leq 0$ then $t$ is degenerate in the sense that $\guard{t}(\Bterm)$ is always $\assmp$-valid, and $\forall x: \guard{t}. \ \guard{t'}(x \cap \Bterm)$ is never $\assmp$-valid unless $t'$ is also degenerate. \iflong
(Note that checking if a threshold is degenerate amounts to a (non) entailment check in BAPA, which is decidable.)
\fi

\iflong
We observe that we can
\begin{inparaenum}[(i)]
\item push conjunctions outside of formulas (since $\forall$ distributes over $\wedge$), and assuming non-degeneracy,
\item ignore terms of the form $x^c$ as, under the assumption that $\thresholds$ is non-degenerate, they will not appear in $\assmp$-valid formulas.
\end{inparaenum}
These observations are formalized by the following lemma.
\fi

\begin{lemma} \label{lem:simple}
If $\thresholds$ is $\assmp$-feasible and $\assmp$-non-degenerate, then for every $\assmp$-valid $\varphi$ in \lang, there exist $\varphi_1,\ldots,\varphi_m$ s.t.
$\varphi \equiv \bigwedge^m_{i=1}\varphi_i$ and for every $1 \leq i \leq m$, $\varphi_i$ is of the form:
\begin{align*}
    \forall x_1:\guard{t_1} \ldots  \forall x_q:\guard{t_q}.\ \guard{t}(x_1 \cap \ldots \cap x_q \cap a_{1} \ldots \cap a_{k})
\end{align*}
where $q+k>0$, $t_1,\ldots,t_q \in \thresholds$, $t \in \thresholdsEx$, $a_1,\ldots,a_k \in \hat{\ParamSet}$, and the $a_i$'s are distinct.
\end{lemma}
We refer to $\varphi_i$ of the form above as \emph{simple}, and refer to $\guard{t}$ as its \emph{atomic guard}.

\ifproofs 

\begin{proof}[sketch]
Let $\varphi$ be a formula in \lang such that $\assmp \models \varphi$. Since universal quantification distributes over conjunction, there exist $\varphi_1,\ldots,\varphi_m$ such that
$\varphi \equiv \bigwedge^m_{i=1}\varphi_i$ and for every $1 \leq i \leq m$, the formula $\varphi_i$ is conjunction-free, i.e., $\varphi_i$ is a quantified atomic formula.

Since $\varphi$ is $\assmp$-valid, each of the $\varphi_i$'s is $\assmp$-valid as well.
This means, because $\thresholds$ is $\assmp$-non-degenerate, that the term $\emptyset$ cannot appear in the formula.
Furthermore, the term $\emptyset^c$ as well as duplicated instances of terms in $\varphi_i$ can be omitted, resulting in an equivalent formula.
As a result, the essence of the proof is to show that $x^c$ will not appear in $\assmp$-valid formulas.

Consider such $\varphi_i$ with atomic guard $t_b$ and a universally quantified variable $x$ with guard $t_a$.
Recall that $\thresholds$ is $\assmp$-non-degenerate and let $\structbapa = (\Dom, \Int)$ be such that $\structbapa \models \assmp$ but $\structbapa \not \models t_b \leq 0$, i.e., $\Int(t_b) > 0$.  Since $\thresholds$ is also $\assmp$-feasible, we have that in particular $\Int(t_a) \leq \Int(\nP)$. Therefore, taking a valuation where $\Val(x)=\Int(\emptyset^c)$ shows that $\structbapa \models \exists x.\ \card{x}\geq t_a \land \card{x^c} < t_b$, hence $\forall x: \guard{t_a} \ldots  .\ \guard{t_b}(x^c \cap \ldots)$ is not $\assmp$-valid.
\qed
\end{proof}

\fi

By \Cref{lem:simple}, it suffices to generate all \emph{simple} $\assmp$-valid formulas. Next, we show that this can be done more efficiently by pruning the search space based on a subsumption relation that is checked syntactically avoiding $\assmp$-validity checks.

\begin{definition}[Subsumption] \label{def:threshold-subsumption}
    For every $h_1, h_2 \in \thresholdsEx \cup \hat{\ParamSet}$, we denote $h_1 \bapaSubvalid h_2$ if one of the following holds:
    \begin{inparaenum}[(1)]
        \item $h_1 = h_2$, or
\item $h_1, h_2 \in \thresholdsEx$ and $\assmp \models h_1 \geq h_2$, or
\item $h_1 \in \hat{\ParamSet}$, $h_2 \in \thresholdsEx$ and $\assmp \models \card{h_1} \geq h_2$.
    \end{inparaenum}
\end{definition}
\iflong
When $h_1,h_2 \in \thresholdsEx$, $h_1 \bapaSubvalid h_2$ means that $\assmp \models \forall x:\guard{h_1}.\ \guard{h_2}(x)$, and when $h_1 \in \hat{\ParamSet}$ and $h_2 \in \thresholdsEx$, $h_1 \bapaSubvalid h_2$ means that $\assmp \models \guard{h_2}(h_1)$.
\fi
We  lift the relation $\bapaSubvalid$ to \iflong a partial order over \fi simple formulas:

\begin{definition}
    Given simple formulas
    \begin{align*}
    \alpha = &\forall x_1:\guard{h_1} \ldots \forall x_q: \guard{h_q}.\ \guard{t}(x_1 \cap \ldots \cap x_q \cap h_{q+1} \ldots \cap h_{k})\\
    \beta = & \forall x_1:\guard{h'_1} \ldots \forall x_{q'}:\guard{h'_{q'}}.\ \guard{t'}(x_1 \cap \ldots \cap x_{q'} \cap h'_{q'+1} \ldots \cap h'_{k'})
    \end{align*}
    we say that $\alpha \bapaSubvalid \beta$ if
    \begin{inparaenum}[(i)]
        \item $t \bapaSubvalid t'$, and
\item there exists an injective function $f:\{1, \ldots, k'\} \rightarrow \{1, \ldots, k\}$ s.t. for any $1 \leq i \leq k'$ it holds that $h'_i \bapaSubvalid h_{f(i)}$.
    \end{inparaenum}
\end{definition}

\begin{lemma}
\label{lem:subsumption}
Let $\alpha, \beta$ be simple formulas such that $\alpha \bapaSubvalid \beta$. If
$\assmp \models \alpha$ then $\assmp \models \beta$.
\end{lemma}

\begin{corollary} \label{cor:validityMaxQuantifier}
If no simple formula with $q$ quantifiers is $\assmp$-valid then no simple formula with more than $q$ quantifiers is $\assmp$-valid.
\end{corollary}

\noindent \Cref{fig:axiom-synt-alg} depicts \ouralg that generates all $\assmp$-valid simple formulas, relying on \Cref{lem:subsumption}.
\iflong
\Cref{fig:axiom-synt-alg} is designed to emit all $\assmp$-valid simple formulas, including ones that are subsumed or entailed by others, since BAPA entailment does not imply entailment between the FO translations\footnote{It is worth mentioning that if the FO formulas that encode subsumption over $\thresholdsEx \cup \hat{\ParamSet}$ (\Cref{def:threshold-subsumption}) are included in the FO translations (these are the formulas $\ToFOL(\forall x:\guard{h}.\ \guard{t}(x))$ for $h \bapaSubvalid t$ and $\ToFOL(\guard{t}(a))$ for $a \bapaSubvalid t$ where $h,t \in \thresholdsEx$ and $a\in \hat{\ParamSet}$), then subsumption between simple formulas does actually imply entailment between the FO translations.
}.
\fi
\ouralg uses a na\"ive search strategy;  different strategies can be used (e.g.~\cite{marco}).
Based on \Cref{cor:validityMaxQuantifier}, \ouralg terminates if for some number of quantifiers no $\assmp$-valid formula is discovered.

\begin{figure}[t]
\begin{algorithm}[H]

\KwIn{$\ParamSet$, $\thresholds$, $\assmp$}
set checked\_true = checked\_false =  [ ] \;
\ForEach {$q = 0,1,\ldots$} {
    \ForEach {simple formula $\varphi$ over $\thresholds$ and $\ParamSet$ with $q$ quantifiers} {
        \lIf {exists $\psi \in$ checked\_true s.t. $\psi \bapaSubvalid \varphi$ } {
            yield $\varphi$ }
        \lElseIf {exists $\psi \in$ checked\_false s.t. $\varphi \bapaSubvalid \psi$ } {
            continue }
        \lElseIf {$\assmp \models \varphi$} {
              yield $\varphi$ ; add $\varphi$ to checked\_true }
        \lElse {
            add $\varphi$ to checked\_false }
    }
    \lIf {no formulas were added to checked\_true} {
        terminate }
}
\caption{\label{fig:axiom-synt-alg} Algorithm for Inferring Intersection Properties (\ITA)}
    \end{algorithm}
\end{figure}

\begin{lemma}[Soundness] \label{lem:soundness}
Every formula $\varphi$ that is returned by the algorithm is $\assmp$-valid.
\end{lemma}
\vspace{-0.5cm}
\begin{lemma}[Completeness] \label{lem:completness}
If $\thresholds$ is $\assmp$-feasible and $\assmp$-non-degenerate,
then for every $\assmp$-valid \lang formula $\varphi$ there exist $\varphi_1\ldots \varphi_m$ s.t.
$\varphi \equiv \bigwedge_{i=1}^m \varphi_i$ and
$\ITA$ yields every $\varphi_i$.
\end{lemma}
\ifproofs
The proof to \Cref{lem:soundness} follows directly from \Cref{lem:subsumption}.
Proof \Cref{lem:completness} is obtained by using \Cref{lem:simple,lem:subsumption}, and \Cref{cor:validityMaxQuantifier}.
\fi
Next, we characterize the cases in which there are finitely many $\assmp$-valid \lang formulas, up to equivalence, and thus, \ouralg is guaranteed to terminate.

\iflong
\begin{definition} \label{def:sanity}
$\thresholds$ is \emph{$\assmp$-sane} if for every $t_1, t_2 \in \thresholds$,  $\assmp \not\models t_1 \leq 0 \lor t_2 > \nP - 1$.
\end{definition}
Sanity implies that no threshold in $\thresholds$ is equivalent to $0$ or to $\nP$ under $\assmp$ (in particular, this implies non-degeneracy). In fact, it captures a stronger requirement that for every pair of thresholds, there is a model of $\assmp$ in which one of them is not interpreted as $0$ and the other is not interpreted as $\nP$.
\fi

\begin{theorem}\label{thrm:finiteConjunctiveFragment}
    Assume that $\thresholds$ is $\assmp$-feasible.
    Then the following conditions are equivalent:
    \iflong
    \begin{enumerate}
    \else
    \begin{inparaenum}[(1)]
    \fi
    \item\label{item:simple} There are finitely many $\assmp$-valid simple formulas.
    \item\label{item:ipf} There are finitely many $\assmp$-valid \lang formulas, up to equivalence.
    \item \label{equ:terminationCondition} $\thresholds$ is \emph{$\assmp$-sane}.
     \iflong
    \end{enumerate}
    \else
    \end{inparaenum}
    \fi
\end{theorem}

\ifproofs 

\noindent The proof that (\ref{equ:terminationCondition}) implies (\ref{item:simple})
uses the following lemma.

\begin{lemma}\label{lem:terminationHelper}
    If $\thresholds$ is $\assmp$-sane, then for every $t_a\in \thresholds$ and $t_b \in \thresholdsEx$, there exists  a number $Q_{t_a,t_b}$ s.t. for every $q \geq Q_{t_a,t_b}$,
$
        \assmp \not\models \forall x_1 : \guard{t_a}.  \ldots  x_{q} : \guard{t_a}. \guard{t_b}(x_1 \cap  \ldots  \cap x_{q}).
    $
\end{lemma}

\begin{proof}[of \Cref{lem:terminationHelper}]
    Let $t_a \in \thresholds,t_b\in\thresholdsEx$ be arbitrary guards. As $\thresholds$ is $\assmp$-sane, there is a structure $\structbapa = (\Dom, \Int)$ such that $\structbapa \models \assmp$ and $\structbapa \models t_a \leq \nP-1 \land t_b > 0$ if $t_b \in \thresholds$, or $\structbapa \models t_a \leq \nP-1 \land t_a > 0$ otherwise. Either way, we have that $\Int(t_b) > 0$ and $\Int(t_a) \leq \Int(\nP) - 1$.

    We define $Q_{t_a,t_b} = \Int(\nP)$. Wlog the universal set of $\structbapa$ is $U = \{e_1 \ldots e_{Q_{t_a,t_b}}\}$. We define a valuation $\Val$ for $\structbapa$ in which for every $1 \leq i \leq Q_{t_a,t_b}$ we set $\Val(x_i) = E_i = U \setminus \{e_i\}$. We then get that
$|\Val(x_i)|\ge \Int(t_a)$ for all $1\le i\le Q_{t_a,t_b}$, and $|\cap_{1\le i\le Q_{t_a,t_b}} \Val(x_i)| = 0 < \Int(t_b)$.
Hence, $
        \assmp \not\models \forall x_1 : \guard{t_a}.  \ldots  x_{q} : \guard{t_a}. \guard{t_b}(x_1 \cap  \ldots  \cap x_{q}).
    $
    By \Cref{lem:subsumption} this also holds for any $q \geq Q_{t_a,t_b}$.
\qed
    \end{proof}

\begin{proof}[of \Cref{thrm:finiteConjunctiveFragment}]
By \Cref{lem:simple}, (\ref{item:simple}) implies (\ref{item:ipf}).
Then (\ref{item:ipf}) implies (\ref{equ:terminationCondition}) since whenever  $\assmp \models t_a > \nP-1 \lor t_b \leq 0$, then any formula of the form $\forall x_1:\guard{t_a}.  \ldots  \forall x_q:\guard{t_a}. \guard{t_b}(x_1 \cap  \ldots  \cap x_q)$ is $\assmp$-valid.
Finally we show that~(\ref{equ:terminationCondition}) implies~(\ref{item:simple}).  Let $Q = \sum_{t_a \in \thresholds} \max_{t_b \in \thresholdsEx} {Q_{t_a,t_b}}$, where $Q_{t_a,t_b}$ is the number defined by \Cref{lem:terminationHelper}.
    Then every $\assmp$-valid simple \lang formula $\alpha$ must have less than $Q$ quantifiers, as otherwise there exists $t_a \in \thresholds$ such that at least
    $\max_{t_b \in \thresholdsEx} {Q_{t_a,t_b}}$ quantifiers have guard $\guard{t_a}$. Suppose that $\guard{t}$ is the atomic guard of $\alpha$, then this means that
    at least $Q_{t_a,t}$ quantifiers have guard $\guard{t_a}$, and thus, from \Cref{lem:subsumption} and \Cref{lem:terminationHelper}, we get that $\assmp \not\models \alpha$.
\qed
\end{proof}

\fi

\begin{corollary}[Termination]
   If $\thresholds$ is $\assmp$-feasible and $\assmp$-sane,
\ouralg terminates.
\end{corollary}

\subsection{From \lang to Axioms in EPR}\label{sec:IPFAcyclicity}

The set of simple formulas generated by \ouralg, $\IPFassmp$, is translated to FOL axioms as described in \Cref{sec:translation-to-fol}.
Next, we show how to ensure that the quantifier alternation graph (\Cref{sec:qa-graph}) of $\ToFOL(\IPFassmp)$ is acyclic.
\iflong
\begin{observation} \label{observation}
A simple formula $\varphi$ with atomic guard $\guard{t}$  induces quantifier alternation edges in $\stgraph{\ToFOL(\varphi)}$ from the threshold sorts of the guards of its universal quantifiers (these are thresholds in $\thresholds$) to the threshold sort of $t$ if $t \in \thresholds$ or to the $\snode$ sort if $t = 1$. If $t = \nP$, no quantifier alternation edges are induced by $\varphi$.
\end{observation}
\fi
Therefore, from \Cref{lem:subsumption}, for a $\assmp$-valid $\varphi$, cycles in $\stgraph{\ToFOL(\varphi)}$ may only occur if they occur in the graph obtained by $\bapaSubvalid$.
Furthermore, if $\stgraph{\ToFOL(\varphi)}$ is not acyclic, then the atomic guard must be equal to one of the quantifier guards. We refer to such a formula as a \emph{cyclic formula}.
We show that, under the following assumption, we can eliminate all cyclic formulas from $\IPFassmp$.

\begin{definition}
$\thresholds$ is \emph{$\assmp$-acyclic} if for every $t_1,t_2 \in \thresholds$, if $\assmp \models t_1 = t_2$ then $t_1 = t_2$.
\end{definition}

Intuitively, if $\thresholds$ is not $\assmp$-acyclic, then it has (at least) two ``equivalent'' thresholds, making one of them redundant. If that is the case, we can alter the protocol and its proof so that one of these guards is eliminated and the other one is used instead.

\iflong
To facilitate elimination of cyclic formulas, we also need to strengthen the sanity requirement (\Cref{def:sanity}) to refer to the set parameters as well:
\begin{definition}
 $(\thresholds,\ParamSet)$ is \emph{$\assmp$-sane} if $\thresholds$ is $\assmp$-sane and, in addition, for every $t \in \thresholds$ and $a \in \hat{\ParamSet}$,
 $\assmp \not\models t \leq 0 \lor \card{a} = \nP$.
\end{definition}
\fi

\begin{theorem}\label{thrm:cycles}
\iflong
Assume that $\thresholds$ is $\assmp$-feasible and $\assmp$-acyclic and that $(\thresholds,\ParamSet)$ is $\assmp$-sane.
Let $\IPFassmp$ be the set of simple $\assmp$-valid formulas returned by \ouralg, and let
$\IPFassmp' = \{ \varphi\in \IPFassmp \mid \varphi \mbox{ is not cyclic}\}$.
Then the verification conditions of the first-order transition system with axioms $\ToFOL(\IPFassmp)$ are valid  if and only if they are valid with axioms $\ToFOL(\IPFassmp')$.
Further, the quantifier alternation graph of $\ToFOL(\IPFassmp')$ is acyclic.
\fi
\end{theorem}

\ifproofs 
To prove the theorem we first prove the following lemma:

\begin{lemma} \label{lem:cycles}
If $\thresholds$ is $\assmp$-feasible and $(\thresholds,\ParamSet)$ is $\assmp$-sane, then for every $\assmp$-valid simple formula $\varphi$, if
$\varphi$ is cyclic, then $\ToFOL(\varphi) \equiv \true$.
\end{lemma}

\begin{proof}[sketch]
Let
$\varphi= \forall x_1:\guard{t_1} \ldots  \forall x_q:\guard{t_q}.\ \guard{t}(x_1 \cap \ldots \cap x_q \cap a_{1} \ldots \cap a_{k})$.
Following \Cref{observation}, because $\varphi$ is cyclic, there must exist $t_i$ such that $t_i = t$.
Assume without loss of generally that this is $t_1$ (the threshold associated with the first quantifier). If $q > 1$, then we get that $\varphi \bapaSubvalid \forall x_1 : \guard{t}.\ \forall x_2 : \guard{t_2}. \guard{t}(x_1 \cap x_2)$.
Because $\thresholds$ is $\assmp$-sane, there is $\structbapa=(\Dom,\Int)$ such that $\structbapa \models \assmp$ but $\structbapa \models t > 0$ and $\structbapa \models t_2 \leq \nP - 1$.
Because $\thresholds$ is $\assmp$-feasible, a valuation $\Val$ over $\structbapa$ exists such that
$\card{\Val(x_1)} = \lceil\Int(t)\rceil >0$ (this is well-defined since $\assmp$-feasibility ensures that $\lceil\Int(t)\rceil \leq \Int(\nP)$), $\card{\Val(x_2)} = \max\{\lceil\Int(t_2)\rceil,0\} < \nP$, hence the guards of the quantifiers are satisfied,
but $\card{\Val(x_1) \cap \Val(x_2)} < \Int(t)$, from which we conclude that $\varphi$ is not $\assmp$-valid. This means that $q=1$.
Similarly, assume $k > 0$, then we have that $\varphi \bapaSubvalid \forall x : \guard{t}.\ \guard{t}(x \cap a_1)$, and because $(\thresholds,\ParamSet)$ is $\assmp$-sane (which ensures that there is $\structbapa$ such that $\structbapa \models \assmp$ but $\structbapa \models t > 0$ and $\structbapa \models \card{a_1} < \nP$), we again conclude that $\varphi$ is not $\assmp$-valid. This means that $k=0$. The conclusion is that $\varphi = \forall x : \guard{t}.\ \guard{t}(x)$, for which $\ToFOL(\varphi) \equiv \true$.
\end{proof}

\paragraph{Proof of \Cref{thrm:cycles}.}
Let $\IPFassmp$ be the set of all $\assmp$-valid \lang formulas, and let $\IPFassmp'$ be defined as in the theorem. \Cref{lem:cycles} shows that $\ToFOL(\IPFassmp') \models \ToFOL(\IPFassmp)$ as required.
Assume $\stgraph{\ToFOL(\IPFassmp')}$ is cyclic with a cycle $\alpha$.
Following \Cref{observation}, the sort $\snode$ cannot be in $\alpha$.
Let $t_1, t_2 \in \thresholds$ be such that $\squorumof{t_1}, \squorumof{t_2} \in \alpha$. Because $\bapaSubvalid$ is transitive, we get that $t_1 \bapaSubvalid t_2$ and $t_2 \bapaSubvalid t_1$, which by definition of $\bapaSubvalid$ means that $\assmp \models t_1 = t_2$, and because $\thresholds$ is $\assmp$-acyclic, we get that $t_1 = t_2$, i.e., $\alpha$ corresponds to a self loop over a single sort $\squorumof{t_1}$. By \Cref{observation}, a self loop may only arise from a cyclic formula. Hence, each cycle in $\stgraph{\ToFOL(\IPFassmp')}$ is induced by a single cyclic formula. As $\IPFassmp'$ contains no cyclic formulas, we conclude that $\stgraph{\ToFOL(\IPFassmp')}$ contains no cycles.
\qed
\medskip

\fi

\subsection{Finding Minimal Properties Required for a Protocol}
\label{sec:impl}

If $\IPFassmp$ consists of \emph{all} acyclic $\assmp$-valid \lang formulas returned by \ouralg,   using $\ToFOL(\IPFassmp)$ as FO axioms leads to divergence of the verifier. To overcome this, we propose two variants.

\paragraph{Minimal Equivalent $\IPFassmp_{min}$.}
Some of the formulas in $\ToFOL(\IPFassmp)$ are implied by others, making them redundant.
We remove such formulas using a greedy procedure that for every
$\varphi_i \in \IPFassmp$,   checks whether $\ToFOL(\IPFassmp \setminus \{\varphi_i\}) \models \ToFOL(\varphi_i)$, and if so, removes $\varphi_i$ from $\IPFassmp$.
Note that if $\stgraph{\ToFOL(\IPFassmp)}$ is acyclic, the check translates to (un)satisfiability in EPR. 

This procedure results in $\IPFassmp_{min} \subseteq \IPFassmp$ s.t. $\ToFOL(\IPFassmp_{min}) \models \ToFOL(\IPFassmp)$ and no strict subset of $\IPFassmp_{min}$ satisfies this condition.
That is, $\IPFassmp_{min}$ is a local minimum for that property.

\paragraph{Interpolant $\IPFassmp_{int}$.}

There may exist $\IPFassmp_{int} \subseteq \IPFassmp$  s.t. 
$\ToFOL(\IPFassmp_{int}) \not \models \ToFOL(\IPFassmp)$
but $\ToFOL(\IPFassmp_{int})$ suffices to prove the first-order VCs, and enables to discharge the VCs more efficiently.
We compute such a set $\IPFassmp_{int}$ iteratively. Initially, $\IPFassmp_{int}= \emptyset$. In each iteration,  we check the VCs.
If a counterexample to induction (CTI) is found, we add to $\IPFassmp_{int}$ a formula from $\IPFassmp$ not satisfied by the CTI.
In this approach, $\IPFassmp$ is not pre-computed. Instead, \ouralg is invoked lazily to generate candidate formulas in reaction to CTIs.

  \section{Evaluation}

\iflong

We evaluate the approach by verifying several challenging threshold-based distributed protocols that use sophisticated thresholds:
Bosco~\cite{SongR08} (presented in~\Cref{sec:exampleBosco}), Hybrid
Reliable Broadcast~\cite{ST87:abc}, and Byzantine Fast
Paxos~\cite{lamport2009fast}.

\subsection{Protocols}
\label{sec:protocols}

\subsubsection{Bosco}

Bosco was explained in detail in \Cref{sec:exampleBosco}. We verified
it under 3 different resilience conditions.
The condition $\nP > 3\tP$ is required in order to ensure correctness.
The condition $\nP > 5\tP$ allows Bosco to guarantee that if there are
no Byzantine processes and all processes have the same input value,
then every processor would reach a decision in a single network step
(\emph{weakly one-step}). If $\nP > 7\tP$ then Bosco ensures that even
if there are some faulty processes, when all non-faulty processes
start with the same initial value, they would reach a decision within
a single network step (\emph{strongly one-step}).
To evaluate our approach, we verify the safety and liveness
(using the liveness to safety reduction presented in~\cite{PadonHLPSS18}) of Bosco.

\subsubsection{Hybrid Reliable Broadcast} We consider the asynchronous reliable
broadcast from~\cite{ST87:abc}, that is designed to
tolerate~$\tP_b<\frac{\nP}{3}$ Byzantine faulty processes.  Hybrid reliable
broadcast, in its extended version, tolerates four different types of faults
(namely Byzantine faults, symmetric faults, clean crash, and crash faults),
with associated constants denoted in \Cref{fig:both} by $f_b$, $f_s$,
$f_c$, and $f_i$, respectively.  The protocol constitutes the core of the clock
synchronization algorithm presented in~\cite{WS07:DC}.  Interestingly,
the protocol is correct under several different threshold conditions~\cite{LKWB17:opodis}.
Nevertheless, the threshold intersection properties are same in all cases,
which confirms that we capture the essence of thresholds,
independently of their arithmetic representation. We verify the safety and liveness of Hybrid
Reliable Broadcast.

\subsubsection{Byzantine Fast Paxos}

Byzantine Fast Paxos~\cite{lamport2009fast} is a fast-learning~\cite{lamportLowerBoundsAsynchronous2003,lamport_lower_2006}
Byzantine fault-tolerant consensus protocols for an
asynchronous system equipped with a leader-election module; 
``fast-learning'' means that under favourable timing, a decision can be reached in a single
round of communication (when each node is simultaneously a proposer, acceptor,
and learner) despite $q$ crash failures.  Moreover, Byzantine Fast Paxos is
optimal in the sense that, in general, no fast-learning protocol can improve
the bounds on $n$~\cite{lamportLowerBoundsAsynchronous2003}. 
Byzantine Fast
Paxos is safe if at most $t$ Byzantine failures occur in a system of $n$ nodes
where $n>3t+2q$ for some parameter $q\leq t$ (Lamport~\cite{lamport2009fast}
assumes that at most $m$ out of the $t$ failures are Byzantine; we only
consider the case $m=t$). We verified the safety of Byzantine Fast
Paxos, i.e., agreement.

Note that Byzantine fault-tolerant fast-learning consensus protocols are notoriously
tricky to develop and verify.  For example, two such algorithms,
Zyzzyva~\cite{KotlaZyzzyvaSpeculativeByzantine2007} and
FaB~\cite{MartinFastByzantineConsensus2006}, were recently revealed
incorrect~\cite{abrahamRevisitingFastPractical2017} despite having been
published at major systems conferences. The verification of Fast Byzantine
Paxos therefore shows that the methodology presented in this paper is able to
handle some of the most intricate distributed protocols.

\fi

\iflong
\subsection{Implementation}
\fi

We implemented both algorithms described in \Cref{sec:impl}.
\ouralgeager eagerly constructs $\IPFassmp$ by running
\ouralg, and then uses EPR reasoning to remove redundant formulas 
(whose FO representation is implied by 
the FO representation of others). To reduce the number of EPR validity checks used during this minimization step, we implemented an optimization that allows us to prove redundancy of \lang formulas internally based on an extension of the notion of subsumption from \Cref{sec:axiom-synt}. \ouralglazy computes a subset of $\IPFassmp$ while using \ouralg in a lazy fashion, guided by CTIs obtained from
attempting to verify the FO transition system.
Our implementations use CVC4 to discharge BAPA queries, and Z3 to
discharge EPR queries. Verification of first-order transition systems
is performed using Ivy, which internally uses Z3 as well. All
experiments reported were performed on a laptop running 64-bit Windows
10, with a Core-i5 2.2 GHz CPU, using Z3 version 4.8.4, CVC4 version
1.7, and the latest version of Ivy. 

\iflong
\subsection{Results}
\fi

\Cref{fig:both} lists the protocols we verified and the details of the
evaluation. Each experiment was repeated 10 times,
and we will report the mean time ($\mu$) and standard deviation ($\sigma$).
The figure's caption explains the presented information,
and we discuss the results below.

\subsubsection{\ouralgeager}

For all protocols, running \ouralg took less than 1 minute (column $\mathbf{t_C}$), and
generated all $\assmp$-valid simple \lang formulas. We observe that
for most formulas, (in)validity is deduced from other formulas by subsumption,
and less than 2\%--5\% of the formulas are actually checked using a
BAPA query.
With the optimization of the redundancy check, minimization of the set is performed in negligible time.
The resulting set, $\IPFassmp_\textsc{Eager}$, contains 3--5 formulas, compared
to 39--79 before minimization.

Due to the optimization described in \Cref{sec:language} for the BAPA validity queries, the number of quantifiers in the \lang formulas that are checked by \ouralg does not affect the
time needed to compute the full $\IPFassmp$.
For example, Bosco under the Strongly One-step resilience condition contains $\assmp$-valid simple
\lang formulas with up to 7 quantifiers (as $\nP > 7\tP$ and $t_1 =
\nP - \tP$), but $\ouralg$ does not take significantly longer to find $\IPFassmp$. Interestingly, in this example the $\assmp$-valid \lang formulas with
more than 3 quantifiers are implied (in FOL) by formulas with at
most 3 quantifiers, as indicated by the fact that these are the only
formulas that remain in $\IPFassmp_\textsc{Eager}^{\text{Bosco Strongly One-step}}$.

\subsubsection{\ouralglazy}

With the lazy approach based on CTIs, the time for finding the set of \lang formulas, $\IPFassmp_\textsc{Lazy}$, is generally longer. This is
because the run time is dominated by calls to Ivy with
FO axioms that are too weak for verifying the protocol. However, the
resulting $\IPFassmp_\textsc{Lazy}$ has a significant benefit: it lets
Ivy prove the protocol much faster compared to using
$\IPFassmp_\textsc{Eager}$. Comparing $\mathbf{t_V}$ in \ouralgeager
vs. \ouralglazy shows that when the former takes a minute, the latter
takes a few seconds, and when the former times out after 1 hour, the
latter terminates, usually in under 1 minute. Comparing the formulas of
$\IPFassmp_\textsc{Eager}$ and $\IPFassmp_\textsc{Lazy}$ reveals the
reason. While the FO translation of both yields EPR formulas, the formulas
resulting from $\IPFassmp_\textsc{Eager}$ contain more quantifiers and
generate much more ground terms, which degrades the performance of
Z3.

Another advantage of the lazy approach is that during the search, it
avoids considering formulas with many quantifiers unless those are
actually needed. Comparing the 3 versions of Bosco we see that \ouralglazy
is not sensitive to the largest number of quantifiers that may appear
in a $\assmp$-valid simple \lang formula. The downside is that \ouralglazy performs
many Ivy checks in order to compute the final $\IPFassmp_\textsc{Lazy}$. The total duration of finding CTIs varies significantly (as demonstrated under the column $\mathbf{t_I}$), in part because it is very sensitive to the CTIs returned by Ivy, which are in turn affected by the random seed used in the heuristics of the underlying solver.

Finally, $\IPFassmp_\textsc{Lazy}$ provides more insight into the
protocol design, since it presents minimal assumptions that are
required for protocol correctness. Thus, it may be useful in designing
and understanding protocols.

\section{Related Work}

\para{Fully automatic verification of threshold-based protocols}
Algorithms modeled as Threshold automata (TA)~\cite{KVW17:IandC} have been
  studied in~\cite{KLVW17:POPL,KLVW17:FMSD}, and verified
  using an automated tool ByMC~\cite{KW18}.
The tool also automatically synthesizes thresholds as arithmetic
  expressions~\cite{LKWB17:opodis}.
Reachability properties of TAs for more general
  thresholds are studied in~\cite{KKW18}.
There have been recent advances in verification of synchronous
  threshold-based algorithms using TAs~\cite{stoilkovska:hal-01925653},
  and of asynchronous randomized algorithms where
  TAs support coin tosses and unboundedly many rounds~\cite{bertrand:hal-01925533}.
Still, this modeling is very restrictive and not as faithful to the pseudo-code
  as our modeling. 

Another approach for full automation is to use sound and incomplete
procedures for deduction and invariant search for logics that combine
quantifiers and set
cardinalities~\cite{gleissenthall_cardinalities_2016,sally}.  However,
distributed systems of the level of complexity we consider here (e.g.,
Byzantine Fast Paxos) are beyond the reach of these techniques.

\para{Verification of distributed protocols using decidable logics} Padon et al.~\cite{PadonMPSS16} introduced an
interactive approach for the safety verification of distributed protocols based
on EPR using the Ivy~\cite{McMillanP18} verification tool.  Later works
extended the approach to more complex protocols~\cite{PadonLSS17}, their
implementations~\cite{TaubeLMPSSWW18}, and liveness
properties~\cite{PadonHLPSS18,PadonHMPSS18}.  Those works verified some
threshold protocols using ad-hoc first-order modeling and axiomatization of
threshold-intersection properties, whereas we develop a systematic methodology. Moreover,
the axioms were not mechanically verified,
except in~\cite{TaubeLMPSSWW18},
where a simple intersection property---intersection of two sets with more than $\frac{\nP}{2}$ nodes---requires a proof by induction over $\nP$. The proof relies on a user provided induction hypothesis that is automatically checked using the FAU decidable fragment~\cite{ge_complete_2009}.
This approach requires user ingenuity even for a simple intersection property,
and we expect that it would not scale to the more complex properties required for e.g.\ Bosco or Fast
Byzantine Paxos.
In contrast, our approach completely automates both verification and inference
of threshold-intersection properties required to verify protocol correctness.

Dragoi et al.~\cite{DragoiHVWZ14} propose a decidable logic supporting
cardinalities, uninterpreted functions, and universal quantifiers for
verifying consensus algorithms expressed in the partially synchronous
Heard-Of Model. As in this paper, the user is expected to provide an
inductive invariant. The PSync framework~\cite{dragoi_psync:_2016}
extends the approach to protocol implementations.  Compared to our
approach, the approach of Dragoi et al. is less flexible due to the
specialized logic used and the restrictions of the Heard-Of Model.

Our approach decomposes verification into EPR and BAPA.
Piskac~\cite{Piskac:168994} presents a decidable logic that combines BAPA and EPR, with some restrictions.
The verification conditions of the protocols we consider are outside the scope of this fragment since
they include cardinality constraints in the scope of quantifiers.
Furthermore, this logic is not supported by mature solvers.
Instead of looking for a specialized logic per protocol, we rely on a decomposition which allows more flexibility.

Recently,~\cite{DBLP:journals/pacmpl/GleissenthallKB19} presented an
approach for verifying asynchronous algorithms by reduction to
synchronous verification. This technique is largely orthogonal and
complementary to our approach, which is focused on the challenge of
cardinality thresholds.

\para{Verification using interactive theorem provers}
We are not aware of works based on interactive theorem provers that verified protocols with complex thresholds as we do in this work (although doing so is of course possible). However, many works used interactive theorem provers to verify related protocols, e.g., \cite{WilcoxWPTWEA15,SergeyWT18,LiuSL17,RahliGBC15,HawblitzelHKLPR15,RahliVVV18} (the most related protocols use either $\frac{{n}}{{2}}$ or $\frac{{2n}}{{3}}$ as the only thresholds, other protocols do not involve any thresholds). The downside of verification using interactive theorem provers is that it requires tremendous human efforts and skills. For example, the Verdi proof of Raft included ~50,000 lines of proof in Coq for ~500 lines of code \cite{WoosWATEA16}.

 \section{Conclusion}

This paper proposes a new deductive verification approach for
threshold-based distributed protocols by decomposing the verification
problem into two well-established decidable logics, BAPA and EPR, thus allowing greater flexibility
compared to monolithic approaches based on domain-specific,
specialized logics.  The user models
their protocol in EPR, defines the thresholds and resilience
conditions using arithmetic in BAPA, and provides an inductive
invariant. An automatic procedure infers  threshold intersection
properties expressed in \lang that are both (1)~sound w.r.t. the resilience conditions and (2)~sufficient to discharge the VCs.
\iflong
Soundness is automatically checked in
(quantifier-free) BAPA, and the VCs are automatically discharged using EPR.
\fi
Both logics are \iflong decidable, \fi supported by mature
solvers, and allow providing the user with an understandable
counterexample in case verification fails.

\iflong
We evaluate the approach by formally verifying the correctness (both
safety and liveness) of intricate protocols, including notoriously
tricky fast-learning consensus protocols such as Byzantine Fast Paxos.
The experimental results show
\fi
that threshold intersection properties
are inferred in a matter of minutes.  While this may be too slow for
interactive use, we expect improvements such as memoization and
parallelism to provide response times of a few seconds in an
iterative, interactive setting. Another potential future direction is
combining our inference algorithm with automated invariant inference
algorithms.

 \remspace
\paragraph{Acknowledgements.}
We thank
the anonymous referees
for insightful comments which improved this paper.
This publication is part of a project that has received funding from the European Research Council (ERC) under the European Union's Horizon 2020 research and innovation programme (grant agreement No [759102-SVIS] and [787367-PaVeS]).
The research was partially supported by
Len Blavatnik and the Blavatnik Family foundation,
the Blavatnik Interdisciplinary Cyber Research Center, Tel Aviv University,
the Israel Science Foundation (ISF) under grant No. 1810/18,
the United States-Israel Binational Science Foundation (BSF) grant No. 2016260
and the Austrian Science~Fund~(FWF) through Doctoral College LogiCS (W1255-N23).
 \newcommand{\odedfrac}[2]{{{\frac{{#1}}{{#2}}}}}

\newcommand{\statresult}[2]{{{\begin{scriptsize}\begin{array}{l}
\mu(#1)\\
  \sigma(#2)
\end{array}\end{scriptsize}}}}

\newcommand{\odedcti}[2]{{
\begin{array}{c}
{#1}\\
{#2}
\end{array}
}}

\begin{landscape}
\begin{figure}[p]
  \begin{footnotesize}
  \arraycolsep=1.5pt
  \[
  \begin{array}{||l|c|c|c|c|c|c|c|l|l|c|c|c|c|c|c||}
    \hhline{|t:================:t|}
    \multirow{2}{*}{\textbf{Protocol}} &
    \multirow{2}{*}{\bf $\thresholds$} &
    \multirow{2}{*}{\bf $\assmp$} &
    \multicolumn{6}{|c|}{\bf \ouralgeager} &
    \multicolumn{7}{|c||}{\bf \ouralglazy}
    \\
&&&
    \textbf{V} &
    \textbf{I} &
    \textbf{Q} &
    \mathbf{t_C} &
    \mathbf{t_V} &
    \multicolumn{1}{|c|}{\textbf{$\IPFassmp_\textsc{Eager}^{\text{Protocol}}$}} &
\multicolumn{1}{|c|}{\textbf{$\IPFassmp_\textsc{Lazy}^{\text{Protocol}}$}} &
    \textbf{V} &
    \textbf{I} &
    \textbf{CTI} &
    \textbf{Q} &
    \mathbf{t_I} &
    \mathbf{t_V}
    \\
    \hhline{||----------------||}
    \mbox{Bosco}
    &
    \begin{array}{l}
      t_1 {=} \nP-t\\
      t_2 {=} \frac{\nP+3t+1}{2}\\
      t_3 {=} \frac{\nP-t+1}{2}
    \end{array}
    &
    \begin{array}{l}
      \nP > 3t\\
      \card{f} \leq t
    \end{array}
    &
    \odedfrac{23}{39} & \odedfrac{21}{1216} & 6 & 3s & \statresult{12s}{4s} &
    \begin{scriptsize}\begin{array}{l}
      g_1(f^c)\\
      \forall x{:}g_1{,}y{:}g_1{,}z{:}g_2{.}g_3(x {\cap} y {\cap} z)\\
      \forall x{:}g_1{,}y{:}g_2{,}z{:}g_3{.}x {\cap} y {\cap} z {\neq} \emptyset
    \end{array}\end{scriptsize}
    &
    \begin{scriptsize}\begin{array}{l}
      g_1(f^c)\\
      \forall x{:}g_1{,}y{:}g_2{.}g_3(x {\cap} y {\cap} f^c)\\
      \forall x{:}g_2{,}y{:}g_3{.}x {\cap} y {\cap} f^c {\neq} \emptyset\\
    \end{array}\end{scriptsize}
    &
    24 & 6 & 18 & 2 & \statresult{3m}{1m} & \statresult{4s}{0.4s}
    \\
    \hhline{||----------------||}
    \begin{array}{l}
      \mbox{Bosco}\\
      \mbox{Weakly}\\
      \mbox{One-step}
    \end{array}
    &
    ''
    &
    \begin{array}{l}
      \nP > 5t\\
      \card{f} \leq t
    \end{array}
    &
    \odedfrac{16}{51} & \odedfrac{24}{1204} & 6 & 3s & \statresult{13m}{14m}
    &
    \begin{scriptsize}\begin{array}{l}
      \IPFassmp_\textsc{Eager}^{\text{Bosco}}\\
      \forall x{:}g_1{.}g_2(x)\\
    \end{array}\end{scriptsize}
    &
    \begin{scriptsize}\begin{array}{l}
      \IPFassmp_\textsc{Lazy}^{\text{Bosco}} \\
      \forall x{:}g_1{.}g_2(x)\\
    \end{array}\end{scriptsize}
    &
    32 & 7 & 19 & 2 & \statresult{13m}{9m} & \statresult{9s}{2s}
    \\
    \hhline{||----------------||}
    \begin{array}{l}
      \mbox{Bosco}\\
      \mbox{Strongly}\\
      \mbox{One-step}
    \end{array}
    &
    ''
    &
    \begin{array}{l}
      \nP > 7t\\
      \card{f} \leq t
    \end{array}
    &
    \odedfrac{26}{63} & \odedfrac{24}{2407} & 8 & 8s & \text{T.O.}
    &
    \begin{scriptsize}\begin{array}{l}
      \IPFassmp_\textsc{Eager}^{\text{Bosco}} \\
      \forall x{:}g_1{,}y{:}g_1{.}g_2(x {\cap} y)\\
    \end{array}\end{scriptsize}
    &
    \begin{scriptsize}\begin{array}{l}
      \IPFassmp_\textsc{Lazy}^{\text{Bosco}} \\
      \forall x{:}g_1{.}g_2(x {\cap} f^c)\\
    \end{array}\end{scriptsize}
    &
    34 & 9 & 20 & 2 & \statresult{23m}{8m} & \statresult{16s}{13s}
    \\
    \hhline{||----------------||}
    \begin{array}{l}
      \mbox{Hybrid}\\
      \mbox{Reliable}\\
      \mbox{Broadcast}\\
    \end{array}
    &
    \begin{scriptsize}\begin{array}{l}
      t_1 {=} t_a {+} t_s {+} 1\\
      t_2 {=} \nP {-} t_c {-} t_a \\ \hfill {-} t_s {-} t_i
    \end{array}\end{scriptsize}
    &
    \begin{scriptsize}\begin{array}{l}
      \nP > t_c {+} 3t_a {+} \\ \hfill 2t_s {+} 2t_i\\
      |f_x| \leq t_x \\
      f_x {\cap} f_y {=} \emptyset \\ \mbox{for }  x\ne y \\ x{,}y{\in}{\{}a{,}c{,}i{,}s{\}}
    \end{array}\end{scriptsize}
    &
    \odedfrac{25}{63} & \odedfrac{34}{1877} & 2 & 37s & \statresult{35s}{0.3s}
    &
    \begin{array}{l}
      g_2(f^c_c {\cap} f^c_b {\cap} f^c_s {\cap} f^c_i)\\
      \forall x{:}g_1{.}x {\cap} f^c_b {\cap} f^c_s {\neq} \emptyset\\
      \forall x{:}g_2{.}g_1(x {\cap} f^c_b {\cap} f^c_i)\\
    \end{array}
    &
    \begin{array}{l}
\IPFassmp_\textsc{Eager}^{\text{Hybrid Reliable Broadcast}}\\    
    \end{array}
    &
    63 & 15 & 45 & 1 & \statresult{15m}{1.5m} & \statresult{43s}{1s}
    \\
    \hhline{||----------------||}
    \begin{array}{l}
      \mbox{Byzantine}\\
      \mbox{Fast}\\
      \mbox{Paxos}\\
    \end{array}
    &
    \begin{array}{l}
      t_1 {=} \nP{-}t\\
      t_2 {=} \nP{-}q\\
      t_3 {=} \nP{-}2t{-}q\\
      t_4 {=} t{+}1
    \end{array}
    &
    \begin{array}{l}
      \nP > 2q+3t\\
      t \geq q\\
      q \geq 0\\
      |b| \leq t
    \end{array}
    &
    \odedfrac{22}{79} & \odedfrac{44}{3695} & 6 & 6s & \text{T.O.}
    &
    \begin{scriptsize}\begin{array}{l}
      g_1(b^c)\\
      \forall x{:}g_2{.}g_1(x)\\
      \forall x{:}g_1{,}y{:}g_4{.}x {\cap} y {\neq} \emptyset\\
      \forall x{:}g_2{,}y{:}g_3{.}g_4(x {\cap} y)\\
      \forall x{:}g_1{,}y{:}g_1{,}z{:}g_2{.}g_3(x {\cap} y {\cap} z)\\
    \end{array}\end{scriptsize}
    &
    \begin{scriptsize}\begin{array}{l}
      g_1(b^c)\\
      \forall x{:}g_2{.}g_1(x)\\
      \forall x{:}g_3{.}x {\neq} \emptyset\\
      \forall x{:}g_4{.}x {\neq} \emptyset\\
      \forall x{:}g_1{.}g_3(x {\cap} b^c)\\
      \forall x{:}g_1{,}y{:}g_1{.}g_4(x {\cap} y)\\
    \end{array}\end{scriptsize}
    &
    44 & 11 & 19 & 2 & \statresult{36m}{21m} & \statresult{28m}{19m}
    \\
    \hhline{|b:================:b|}
  \end{array}
  \]
  \end{footnotesize}
  \caption{\label{fig:both}Protocols verified using our technique.
  For each protocol, $\thresholds$ is the set of thresholds and $\assmp$ is the resilience condition.
  $\ouralgeager$ lists metrics for the procedure of finding all $\assmp$-valid \lang formulas (taking time $\mathbf{t_C}$), and verifying the transition system using the resulting
  properties (taking time $\mathbf{t_v}$). Obtaining a minimal subset that FO-implies the rest takes negligible time, so we did not include it in the table.
  The properties are given in \textbf{$\IPFassmp_\textsc{Eager}^{\text{Protocol}}$}, where $g_i$ denotes $\guard{t_i}$.
  In addition to the run times, \textbf{V} shows $\odedfrac{c}{v}$, where $c$ is the number of $\assmp$-valid simple formulas that were checked using the BAPA solver (CVC4), and $v$ is the total number of  $\assmp$-valid simple formulas.
  Namely $v-c$ simple formulas were inferred to be valid via subsumption.
  \textbf{I} reports the analogous metric for $\assmp$-invalid simple formulas.
  Finally, \textbf{Q} reports the maximal number of quantifiers considered (for which all formulas were $\assmp$-invalid).
$\ouralglazy$ lists metrics for the procedure of finding a set of $\assmp$-valid \lang formulas sufficient to prove the protocol
  based on counterexamples. The resulting set is listed in \textbf{$\IPFassmp_\textsc{Lazy}^{\text{Protocol}}$},
  and $\mathbf{t_I}$ lists the total Ivy runtime, with the standard deviation specified below.
  \textbf{V} (resp. \textbf{I}) lists the number of $\assmp$-valid (resp. $\assmp$-invalid) simple formulas
  considered before the final set was reached.
  \textbf{CTI} lists the number of counterexample iterations required, and \textbf{Q} lists the
  maximal number of quantifiers of any \lang formula considered.
  Finally, $\mathbf{t_v}$ lists the time required to verify the first-order transition system assuming
  the obtained set of properties.
  T.O. indicates that a time out of 1 hour was reached.
}
\end{figure}
\end{landscape}

\bibliographystyle{splncs04}

\end{document}